\documentclass[twocolumn,10pt]{scrartcl}	

\usepackage[utf8]{inputenc}
\usepackage[english]{babel}
\usepackage{fullpage}
\usepackage[margin=2cm
]{geometry}

\usepackage{amssymb, amsmath, amsthm}

\usepackage{graphicx}
\usepackage{algorithm}
\usepackage{algorithmicx}
\usepackage{algpseudocode}

\usepackage{tikz,wrapfig}
\usetikzlibrary{calc,intersections,through,backgrounds,decorations.pathreplacing}

\usepackage{multirow}

\usepackage{hyperref}
\usepackage{xcolor}
\hypersetup{
	colorlinks,
	linkcolor={red!90!black},
	citecolor={green!80!black},
	urlcolor={blue!80!black}
}

\newcommand{\ignore}[1]{}

\newtheorem{observation}{Observation}
\newtheorem{lemma}{Lemma}
\newtheorem{theorem}{Theorem}
\newtheorem{corollary}{Corollary}

\newcommand{\Z}{\mathbb{Z}}
\newcommand{\R}{\mathbb{R}}
\newcommand{\N}{\mathbb{N}}
\newcommand{\Q}{\mathbb{Q}}
\newcommand{\C}{\mathbb{C}}

\renewcommand{\S}{\mathbb{S}}
\renewcommand{\epsilon}{\varepsilon}

\newcommand{\B}{\mathbf{B}}

\begin{document}
	
	\title{Rational Points on the Unit Sphere: Approximation Complexity and Practical Constructions}
	\subtitle{Improved Analysis\cite{self}}
	
	\author{ Daniel Bahrdt	\thanks{Formale Methoden der Informatik, University of Stuttgart,  Germany, {\{bahrdt, seybold\}@fmi.uni-stuttgart.de}} \and Martin P.  Seybold\footnotemark[1]}
	
	\maketitle
	\subsubsection*{Abstract} 	\vspace{-.4cm}
		Each non-zero point in $\R^d$ identifies a closest point $x$ on the unit sphere $\S^{d-1}$.
		We are interested in computing an $\epsilon$-approximation $y \in \Q^d$ for $x$, that is \emph{exactly} on $\S^{d-1}$ and has low bit size. We revise lower bounds on rational approximations and provide explicit, spherical instances.
		
		We prove that floating-point numbers can only provide trivial solutions to the sphere equation in $\R^2$ and $\R^3$. Moreover, we show how to construct a rational point with denominators of at most $10(d-1)/\varepsilon^2$ for any given $\epsilon \in \left(0,\tfrac 1 8\right]$, improving on a previous result.
		The method further benefits from algorithms for simultaneous Diophantine approximation.
		
		Our open-source implementation and experiments demonstrate the practicality of our approach in the context of massive data sets Geo-referenced by latitude and longitude values.
	\vspace{-.4cm}
	\paragraph*{Keywords}
	{ Diophantine approximation,
		Rational points,  
		Unit sphere,
		Perturbation, 
		Stable geometric constructions}
	
\begin{figure}[h]
	\centering
	\includegraphics[clip,trim=420px 220px 410px 200px,width=\columnwidth]{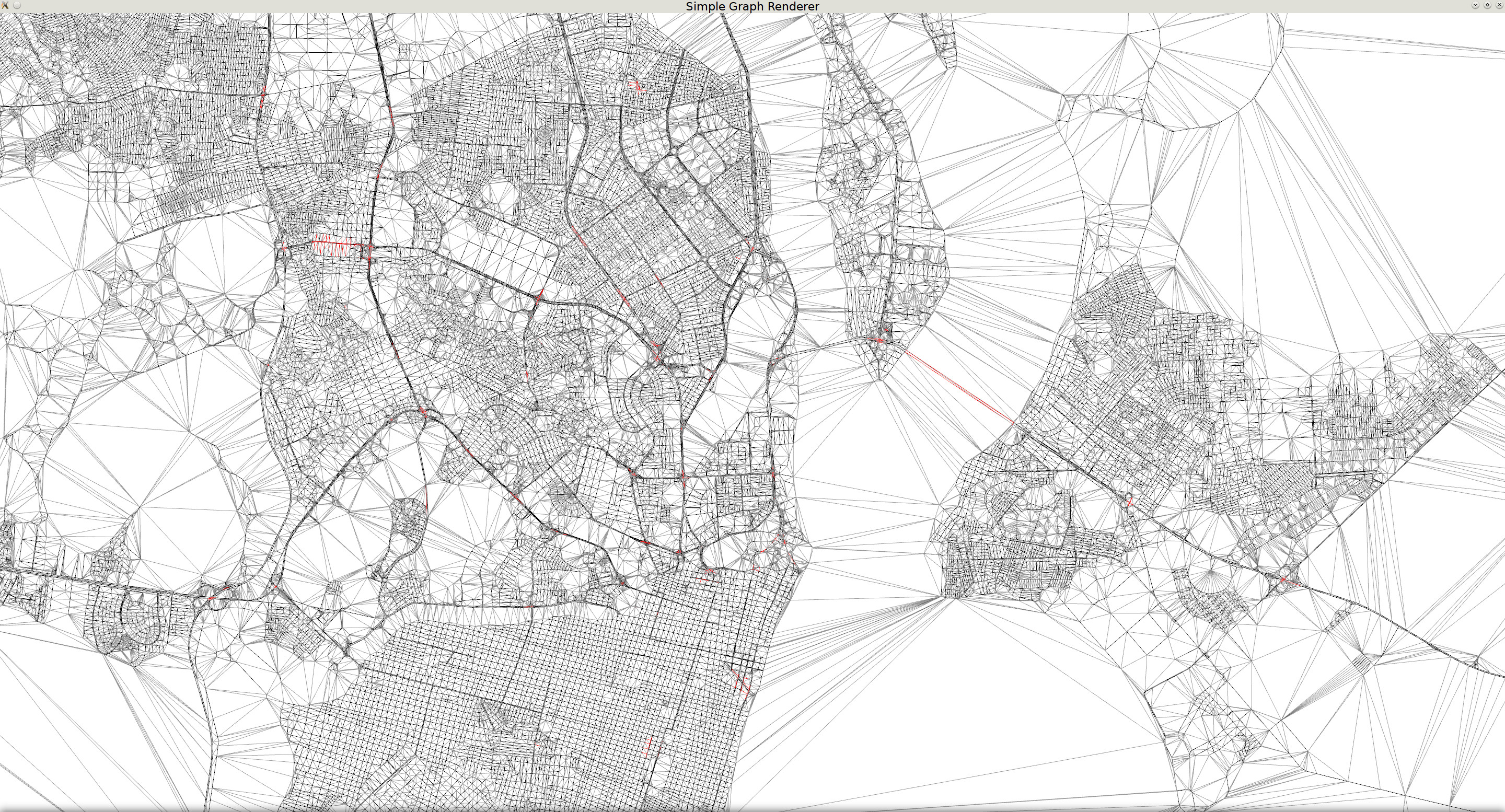}
	\vspace{-.5cm}	
	\caption{Spherical Delaunay triangulation (gray) constrained to contain all line segments (black) of streets in Ecuador and the intersection points of constraints (red).}
	\vspace{-.3cm}
	\label{figMotivation}
\end{figure}

\section{Introduction}
Many mathematical sciences use trigonometric functions in symbolic coordinate transformations to simplify fundamental equations of physics or mathematical systems.
However, rational numbers are dominating in computer processing as they allow for simple storage as well as fast exact and inexact arithmetics (e.g. GMP\cite{Granlund12}, IEEE Float, MPFR\cite{Fousse:2007:MMB:1236463.1236468}).
Therefore problems on spherical surfaces often require to scale a point vector, as in choosing a point  uniform at random\cite{marsaglia1972choosing}, or to evaluate a trigonometric function for a rational angle argument, as in dealing with Geo-referenced data.

A classical theoretical barrier is Niven's theorem\cite{book/niven/irrationalNumbers}, 
 which states that the sole rational values of sine for rational multiplies of $\pi$ are $0, \pm1/2$ and $\pm 1$.
The well known Chebyshev polynomials have roots at these values, hence give rise to representations for these algebraic numbers.
However, arithmetics in a full algebraic number field might well be too demanding for many applications.
For products of sine and cosine, working with Euler's formula on the complex unit circle and Chebyshev polynomials would suffice though.

This manifests in problems of \emph{exact} geometrical computations, since standard methodology relies on Cartesian input\cite{DBLP:journals/jlp/LiPY05}.
Spheres and ellipsoids are common geometric objects and rational solutions to their defining quadratic polynomials are closely related to Diophantine equations of degree $2$.
The famous Pythagorean Triples are known to identify the rational points on the circle $\S^1$.
Moreover, the unit sphere has a dense set of rational points and so do ellipsoids with rational half-axes through scaling.
Spherical coordinates are convenient to reference such Cartesians with angle coordinates and \emph{geo-referenced data} denotes points with rational angles.
Standard approximations of Cartesians do not necessarily fulfill these equations, therefore subsequent algorithmic results can suffer greatly.	
	
This paper focuses on finding rational points \emph{exactly on} the unit sphere $\S^{d-1}=\left\{x \in \R^d ~:~ \sum_{i} x_i^2 = 1\right\}$ with bounded distance to the point $x / \lVert x \rVert_2$ -- its closest point on $\S^{d-1}$.
In this work, $x \in \R^d$ can be given by any finite means that allow to compute a rational approximation to it with arbitrary target precision.
Using rational Cartesian approximations for spherical coordinates, as derived from MPFR, is just one example of such a black-box model.
Moreover, we are interested in calculating rational points on $\S^d$ with small denominators.

\subsection{Related Work} \label{secRelatedWork}
Studies on spherical Delaunay triangulations (SDT), using great-circle segments on the sphere $\S^2$, provide common ways to avoid and deal with the point-on-sphere problem in computational geometry.

The \emph{fragile} approaches
\cite{DBLP:conf/ppam/PrillZ15, gmd-6-1353-2013, 1997-Renka-STRIPACK}
ignore that the input may not be on $\S^2$ and succeed if the results of all predicate evaluations happen to be correct.
Input point arrangements with close proximity or unfortunate locations bring these algorithms to crash, loop or produce erroneous output.
The \emph{quasi-robust} approaches \cite{1979-Brown, 1996-Barber-QHULL} weaken the objective and calculate a Delaunay tessellation in $d$-Simplexes.
Lifting to a $d+1$ convex hull problem is achieved by augmenting a rational coordinate from a quadratic form -- The augmented point exactly meets the (elliptic) paraboloid equation.
However, the output only identifies a SDT if all input points are already on the sphere, otherwise the objectives are distinct. Equally unclear is how to address spherical predicates and spherical constructions.
The \emph{robust} approaches \cite{saalfeldStereoDelaunay} use the circle preserving stereographic projection from $\S^2$ to the plane.
The perturbation to input, for which the output is correct, can be very large as the projection does not preserve distances. 
Furthermore, achieving additional predicates and constructions remains unclear.
The \emph{stable} approaches provide geometric predicates and constructions for points on $\S^2$ by explicitly storing an algebraic number, originating from scaling an ordinary rational approximation to unit length\cite{2009-CastroCLT}.
Algebraic number arithmetics can be avoided for $\S^2$, but exact evaluation relies on specifically tailored predicates \cite{Teillaud2010}, leaving the implementation of new constructions and predicates open.

Kleinbock and Merrill provide methods to quantify the density of rational points on $\S^d$ \cite{Kleinbock2015}, that extend to other manifolds as well.
Recently, Schmutz\cite{Schmutz2008} provided an divide-\&-conquer approach on the sphere equation, using Diophantine approximation by continued fractions, to derive points in ${\Q^d \cap \S^{d-1}}$ for a point on the unit sphere $\S^{d-1}$.
The main theorem bounds the denominators in $\varepsilon$-approximations, under the $\lVert~\rVert_\infty$ norm, with
$
( \sqrt{32} \lceil \log_2 d \rceil / \varepsilon  )^{2\lceil \log_2 d \rceil }
$.
Based on this, rational approximations in the orthogonal group $O(n,\R)$ and in the unitary matrix group $U(n,\C)$ are found. 
This is of particular interest for sweep-line algorithms:
\cite{92-Canny-RatRot} studies finding a rotation matrix with small rationals for a given rational rotation angle of an $2$D arrangement.

\subsection{Contribution}
The strong lower bound on rational approximations to other rational values does not hold for Geo-referenced data considering Niven's theorem.
We derive explicit constants to Liouville's lower bound, for a concrete Geo-referenced point, that is within a factor $2$ of the strong lower bound.
Moreover, we prove that floating-point numbers \emph{cannot} represent Cartesian coordinates of points that are exactly on $\S^1$ or $\S^2$.

We describe how the use of rotation symmetry and approximations with fixed-point numbers suffice to improve on the main theorem of \cite{Schmutz2008}.
We derive rational points \emph{exactly on} $\S^{d-1}$ with denominators of at most
$10(d-1)/\varepsilon^2 $
for any $\varepsilon \in \left(0,\tfrac 1 8\right]$.
Moreover, our method allows for even smaller denominators based on algorithms for simultaneous Diophantine approximations, though a potentially weaker form of approximation would suffice.

The controlled perturbations provided by our method allow exact geometric algorithms on $\S^d$ to rely on rational rather than algebraic numbers -- E.g. enabling convex hull algorithms to efficiently obtain spherical Delaunay triangulations on $\S^d$ and not just Delaunay tessellations.
Moreover, the approach allows for inexact but $\varepsilon$-stable geometric constructions -- E.g. intersections of Great Circle segments.

We demonstrate the quality and effectiveness of the method on several, including one whole-world sized, point sets.
We provide open-source implementations for the method and its application in the case of spherical Delaunay triangulations with intersections of constraints.

\section{Definitions and Tools}
The $2$nd \emph{Chebyshev polynomials} $U_n$ of degree $n$ are in $\Z[X]$, given their recursive definition:
\begin{align*}
U_0(x)&=1 \quad \quad U_1(x)=2x\\
U_{n+1}(x)&=2xU_n(x)-U_{n-1}(x)\quad.
\end{align*}
It is well known \cite{rivlin74}, that the $n$ roots of $U_n$ are exactly the values  
$
	\left\{~\cos \left({\pi k / \left(n+1\right)} \right)~:~k=1,\ldots ,n ~\right\}
$.
Hence the polynomials $U_n$ give rise to algebraic representations for cosine values of rational multiplies of $\pi$.
This is particularly useful in conjunction with classic results on \emph{Diophantine approximations}, that are known since 1844\cite{1851-Liouville}:

\begin{theorem}[Liouville's Lower Bound] \label{thmLiouville}
	For any algebraic $\alpha \in \R$ of degree $n \geq 2$,
there is a positive constant $c(\alpha)> 0$ such that 
	\[
	\Big| \alpha - \frac{p}{q} \Big| \geq \frac{c(\alpha)}{q^n}
	\]
for any $p \in \Z$ and $q \in \N$.
\end{theorem}
Apart from this lower bound on rational approximations, there is another important folklore result on the existence of simultaneous Diophantine approximations.
Such approximations have surprisingly small errors, despite their rather small common denominator.
\begin{theorem}[Dirichlet's Upper Bound] \label{thmDirichlet}
	Let $N \in \N$ and $\alpha \in\R^d$ with $0 \leq \alpha_i \leq 1$.
	There are integers $p \in \Z^d, ~q \in \Z$ with $1 \leq q \leq N$ and 
	$$\left|\alpha_i - \tfrac {p_i}q\right|\leq \frac{1}{q\sqrt[d]{N}}~.$$
\end{theorem}
	(See Appendix \ref{apxProofDirichlet} for a proof.)
For $d=1$, the continued fraction (equivalently the Euclidean) algorithm is famous \cite{hardy54} for finding approximations with $ \left| \alpha - p/q\right| \leq 1/2q^2$.
This spurred the field of number theory to study generalizations of the continued fraction algorithm that come close to Dirichlet's upper bound, but avoid brute-force calculations.
Some more recent methods are discussed in Section \ref{secDenomSize}.

Our approach uses the \emph{Stereographic Projection} in $\R^d$.
Let $p=(0,\dots,0,1) \in \R^d$ be the fixed point for the projection $\tau$, mapping all points of a ray from $p$ to the intersection with the hyperplane $x_d = 0$.
\begin{align*}
\tau : \R^d\setminus(\R^{d-1}\times\{1\}) &\to \R^{d-1}  \\
x &\mapsto \Big(\frac{x_1}{1-x_d}~,~\dots~,~\frac{x_{d-1}}{1-x_d} \Big)
\end{align*}
The surjective mapping $\tau$ is injective as well, when restricted to the domain $\S^{d-1}\setminus\{p\}$.
We further define the mapping $\sigma$, which is
\begin{align*}
\sigma : \R^{d-1} &\to \R^d\setminus\{p\}    \\
x &\mapsto \Big( ~\frac{2x_1}{1+S^2}~,~\dots~,~\frac{2x_{d-1}}{1+S^2}~,~\frac{-1+S^2}{1+S^2}~\Big) 
\end{align*}
where $S^2 = \sum_{j=1}^{d-1}x_j^2 $.
We have  $\operatorname{img} \sigma \subseteq \S^{d-1}$, since 
$$\lVert \sigma(x) \rVert_2^2= \frac{(-1+S^2)^2 + \sum_{i=1}^{d-1} (2x_i)^2}{(1+S^2)^2} = 1\quad.$$
Furthermore, $ x = \tau \circ \sigma( x)$ for all $x\in \R^{d-1}$, since
$$
\left(\tau \circ \sigma \right)_i(x) =
\frac{\frac{2x_i}{1+S^2}}{1-\frac{-1+S^2}{1+S^2}} = \frac{2x_i}{1+S^2+1-S^2} =x_i
$$
holds for all $1 \leq i < d $.
Hence, $\sigma$ and $\tau$ are inverse mappings.
Note that images of rational points remain rational in both mappings, establishing a bijection between rational points in $\R^{d-1}$ and $\S^{d-1}$.

\subsection{Lower Bounds and Instances for Geo-referenced Data on $\S^d$} \label{sectLowerBounds}
It is well known in Diophantine approximation that rational numbers have algebraic degree $1$ and are hard (in the following qualitative sense) to approximate with other rational numbers.
The following folklore observation is an analog to Liouville's lower bound.

\begin{observation}
	For  rational numbers $\tfrac ab \neq \tfrac pq$, we have
	$$
	\left| \frac{a}{b}- \frac{p}{q} \right| =	\left| \frac{aq-bp}{bq} \right| \geq \frac{1}{bq}
	$$
\end{observation}
If $q<b$, we have a lower bound of $1/q^2$ for rational approximations to $\tfrac ab$ with denominators up to $q$.
Pythagorean triples $(x,y,z) \in \N^3$ provide such rational points on $\S^1$, since $(x/z)^2 + (y/z)^2 =1$.
We have a lower bound of $1/z^2$ for approximations with denominators $q<z$.
See Section \ref{secDenomSize} for rational points on $\S^d$ with the same denominator property.

The situation might look different when dealing with Geo-re\-fer\-enced data (rational angle arguments) only.
However, using Chebyshev's polynomials in conjunction with Liouville's lower bound (c.f. Theorem \ref{thmLiouville}) allows to derive explicit constants for Diophantine approximations of $\cos\left(108^\circ\right)$.

Given spherical coordinates, the first coordinate of a point on $\S^{d}$ might well have algebraic values of $r_i= \cos(\tfrac i5\pi)$ for $i \in \{1,2,3,4\}$.
\begin{align*} 
(r_1,r_2,r_3,r_4)
&=\left( \frac{1+\sqrt{5}}{4}, \frac{-1+\sqrt{5}}{4}, \frac{1-\sqrt{5}}{4}, \frac{-1-\sqrt{5}}{4}\right)\\
&\approx( +0.8090, +0.3090, -0.3090, -0.8090)
\end{align*}
Over $\Z[X]$, the polynomial $U_4(x)=16x^4-12x^2+1$ has the irreducible factors
\begin{align*}
U_4(x) = \underbrace{(4x^2 - 2x + 1)}_{=:f(x)}(4x^2 + 2x - 1)
\end{align*}
Since $r_1$ and $r_3$ are the roots of the polynomial $f$, they have algebraic degree $n=2$.

Using Liouville's lower bound for $r_3$, we have for all $\tfrac pq \in \Q$
\[
\Big| r_3 - \frac{p}{q} \Big| \geq \frac{\min\{c_2, \tfrac 1{c_1}\}}{q^n} \quad,
\]
with constants $c_1$ and $c_2$ according to the proof of Liouville's Theo\-rem\cite{1851-Liouville}.
The constants $c_1,c_2>0$ exist, since the polynomial division of $f$ with the linear factor $(x-r_3)$ results in the continuous function $g(x)=(x-r_1)$.
For $c_2 = 1/2 < \sqrt{5}/2$, the interval $I:=[r_3- c_2, r_3 +c_2] \subseteq \R$ is sufficiently small to exclude different roots of $f$ and the inequality
\begin{align*}
\max_{x\in I} \big|g(x)\big| = 	\max_{x\in I} \big|x - r_1\big| < c_1
\end{align*}
is met with a generous choice of $c_1=2$.
This leads to an explicit lower bound on the approximation error to $r_3$ with denominators $q$ of
\[
\left| \cos\left(108^\circ\right) - \frac{p}{q} \right| \geq \frac{1}{2 \cdot q^2} ~.
\]

\section{Results}
Apart from integers, contemporary computing hardware heavily relies on \emph{floating point numbers}.
These are triplets $(s,m,e)$ with $s \in \{0,1\}$, $m \in \{0,\dots,2^l-1\}$ and $e \in \{-2^{k-1}+1,\dots, 2^{k-1}-1\}$.
The IEEE standard for Float is $(l,k)=(23,8)$ and $(52,11)$ for Double.
The rational number described by such a triplet is 
\begin{align*}
\operatorname{val}(s,m,e) = (-1)^s \cdot
\begin{cases}
\dfrac{2^l + m}{2^l}  2^e   							&\hspace{1cm} e > 0  \vspace{.3cm} \\
\dfrac{2^l + m}{2^l}  \dfrac{1}{2^{|e|}} 				&\hspace{1cm} e < 0  \vspace{.3cm} \\
\dfrac{  0 + m}{2^l}  \dfrac{1}{2^{2^{k-1}-2}}      	&\hspace{1cm} e = 0  
\end{cases}
\end{align*}
where the latter case describes `denormalized' numbers.
In each case, the uncanceled rational value has some power of $2$ as the denominator.
Since powers of two are the sole divisors of a $2^i$, the denominator of the canceled rational has to be a power of two, too.
Hence, rational values representable by floating point numbers are a subset of the following set $P$ and fixed-point binary numbers are a subset of $P_i$:
\begin{align*}
\operatorname{img}\operatorname{val} \subseteq \left\{ \frac{z}{2^i} : i \in \N, z \in \Z, z \text{~odd} \right\}=& P \\
\left\{ \frac{z}{2^i} : z \in \Z \right\} = P_i \subseteq &P \quad.
\end{align*}
	
\subsection{Floating Point Numbers are Insufficient}
Fix-point and floating-point arithmetics of modern CPUs work within a subset of rational numbers, in which the denominator is some power of two and the result of each arithmetic operation is `rounded'.
\begin{theorem}\label{thmNoFloats}
	There are only $4$ floating point numbers on $\S^1$ and $6$ on $\S^2$.
\end{theorem}
\begin{proof}
	We show $\S^{d-1} \cap P^d \nsubseteq \{-1,0,1\}^d$ implies ${d \geq 4}$.
	Suppose there is a non-trivial $p \in \S^{d-1}\cap P^d$ with $d$ minimal.
	Let $x_i/2^{e_i}$ denote the canceled fraction of its $i$-th coordinate.
	We have that all $x_i \neq 0$, $x_i$ are odd numbers and all $e_i >0$ (since $p$ is not one of the $2d$ poles and $d$ is minimal).
	
	W.l.o.g. $e_1 \leq e_2 \leq \ldots \leq e_d$.
	We rewrite the sphere equation $1 = \sum_{j=1}^{d} (x_i/2^{e_i})^2$ to
	\begin{align*}
	x_1^2 &= 4^{e_1} - \sum_{j=2}^{d}4^{e_1 - e_j} x_j^2 \quad.
	\end{align*}
	For an odd integer $y$, we have $y^2=(2k+1)^2= 4(k^2+k)+1$, leading to the congruence
	\begin{align*}
	1 \equiv 0 - \sum_{j=2}^{d}\chi_{e_1}(e_j)\quad\mod 4 ~.
	\end{align*}
	Where the characteristic function $\chi_{e_1}(e_j)$ is $1$ for $e_1 = e_j$ and $0$ otherwise.
	For $d\in\{2,3\}$ the right hand side can only have values of $0,-1$ or $-2$, a contradiction.
\end{proof}
Note that theorem \ref{thmNoFloats} translates to spheres with other radii through scaling.
Suppose a sphere in $\R^3$ of radius $2^j$ has a non-trivial solution $y \in P^3$, then $y/2^j \in P^3$ and would be on $\S^2$, too.

\subsection{Snapping to Rational Points}
We now describe how to compute a good rational approximation \emph{exactly on} the unit sphere $\S^{d-1}$.
The input point $x \in \R^d$ can be given by any finite means that allows to compute rational approximations of arbitrary target precision -- E.g. rational approximations of Cartesians for spherical coordinates.
For the input $x$, we denote its closest point on $\S^{d-1}$ with $x/\lVert x \rVert_2$.
The stereographic projection $\tau$ and its inverse mapping $\sigma$ provide $\sigma\left(\tau\left( x/\lVert x\rVert_2 \right) \right)= x/\lVert x\rVert_2$, since the argument is on $\S^{d-1}$.
Instead of determining the value of $\tau$ exactly, we calculate an approximation $y \in \Q^d$ and finally evaluate $\sigma(y)$ under exact, rational arithmetics.
Hence, the result $\sigma(y)$ is exactly on $\S^{d-1}$.
\begin{wrapfigure}{r}{.35\columnwidth}
	\centering
	\begin{tikzpicture}[scale=.13]
	\draw (0,0) circle (10);
	\fill (10,-10) circle (.25) node [right] 			{\small $x$};
	\fill (7,-7) circle (.25) node [left] 				{\small $x / \lVert x \rVert_2$};
	\fill (8,-6) circle (.25) node [right] 				{\small $\sigma(y)$};
	\fill (0,+10) circle (.25) node [above] 			{\small $(0,1)$};
	\fill (4.14213562,0) circle (.25) node [above left] {\small $\tau\left(x/\lVert x \rVert_2\right)~$};
	\fill (5,0) circle (.25) node [above right] 		{\small $y$};
	\draw[-] (-12,0) -- node [above] {} (12,0);	 
	\draw[-,dashed,gray] (0,0) -- node [above] {} (10,-10); 
	\draw[-] (0,10) -- node [above] {} (7.07106781,-7.07106781);
	\draw[-] (0,10) -- node [above] {} (8,-6);
	\end{tikzpicture}
\end{wrapfigure}
The stereographic projection does not preserve distances, leaving it open to bound the approximation error and the size of the resulting denominators.
We use the \emph{rotation symmetry} of the sphere to limit the stretching of $\sigma$ (c.f. Lemma \ref{lemGeoStretching}):
For a non-zero point $x \in \R^d$ we can assume that $i=d$ maximizes $|x_i|$ and $x_d < 0$, otherwise we change the standard orthonormal basis by swapping dimension $i$ and $d$ and using a negative sign for dimension $d$.
Note that such rotations do not change the actual coordinate values.
To keep the \emph{size of denominators} in $\sigma(y)$ small, we use fixed-point arithmetics to determine $y \in \Q^{d-1}$ (c.f. Lemma \ref{lemDenomSize}).

\begin{algorithm}[H]
	\caption{PointToSphere} \label{algoSnapping}
	\begin{enumerate}
		\item[In:] $x \in \R^d,\quad \varepsilon \in \left(0,\tfrac18\right]$
		\item Assert $x_d = \min_i-|x_i|$
		\item Choose $y\in \Q^{d-1}$ with
			$| y_i - \tau_i\left( x / \lVert x\rVert_2 \right) | \leq \frac{\varepsilon}{2\sqrt{d-1}}$
			\label{algoStatementChoose}
		\item Return $\sigma(y) \in \Q^d$.
			\label{algoStatementReturn}
	\end{enumerate}
\end{algorithm}

See Algorithm \ref{algoSnapping} for a precise description.
Note that the rational point $y$ in statement $\ref{algoStatementChoose}$ solely needs to meet the target approximation in the individual coordinates for
$$
\tau_i (x / \lVert x \rVert_2 ) = \frac{x_i}{\lVert x \rVert_2 - x_ d} \quad.
$$
Generally, this can be determined with methods of `approximate expression evaluation' to our target precision\cite{DBLP:journals/jlp/LiPY05}.
If $x$ is an approximation to a geo-referenced point, this denominator is well conditioned for calculations with multi-precision floating-point arithmetics\cite{DBLP:conf/compgeom/BurnikelFS98, Fousse:2007:MMB:1236463.1236468}.
Using exact rational arithmetics for statement $\ref{algoStatementReturn}$, we obtain a rational Cartesian coordinates \emph{on} the unit sphere.

\begin{observation} \label{obsTauInCircle}
	For $d > 1$ and $ x \in \S^{d-1}$ with $x_d = \min_i -|x_i|$, we have
	\[
\lVert	\tau(x) \rVert_2 \leq \sqrt{ \frac{\sqrt{d}-1}{\sqrt{d}+1} } < 1 \quad.
	\]
\end{observation}
\begin{proof}
Using $x_d = \min_i -|x_i|$ and $\sum_{i}x_i^2= 1$, we have the bounds $1/d \leq x_d^2 \leq 1$ and
\begin{align*}
\lVert \tau ( x) \rVert_2^2
&= \frac{\sum_{i=1}^{d-1} x_i^2}{(1-x_d)^2}
 = \frac{1 - x_d^2}{(1-x_d)^2}
 = \frac{1+x_d}{1-x_d}
\leq \frac{1-1/\sqrt{d}}{1+1/\sqrt{d}} ~.
\end{align*}
Where the latter term is in $(0,1)$ for any $d$.\end{proof}
Hence the ($d-1$)-ball $ 
\B^{d-1}_1 = \{ x \in \R^d ~:~ \lVert x \rVert_2 \leq 1\}
$ contains $\tau(x)$.

\subsection{Approximation Quality}
See \cite{self} for an earlier version of this paper with a weaker, but elementary, analysis.

We consider the problem in  the $2$D hyperplane $H_{pyy'}$, defined by two points $y=\sigma(x)$, $y'=\sigma(x')$ on $\S^{d-1}$ and the projection pole $p\in \R^d$.
Given the rotation step in Algorithm \ref{algoSnapping}, the projection plane ${H_0=\{ x \in \R^d: x_d = 0\}}$ separates $p$ and $y,y'$ in $\R^d$ and in $H_{pyy'}$.
Since each ${q \in H_0 \cap S^{d-1}}$ has $\lVert q - p\rVert_2=\sqrt{2}$ (consider $\overline{pq}$ in $H_{0pq}$), 
the circumcircle $C$ of $p,y$ and $y'$ contains exactly two of these points.
Hence, the line of $H_{pyy'}\cap H_0$ is orthogonal to the circumcircle's diameter through $p$.
Moreover, the circles diameter is in $[\sqrt{2},2]$.
We denote with $x$ the point that is closer to $p$ in $H_{pyy'}$, meaning ${\lVert x\rVert_2 \leq \lVert x'\rVert_2}$.
Note that $x'$ and $x$ can be on the same or opposite circumcircle halves.

\begin{center}
	\begin{tikzpicture}[scale=.3]
	\coordinate (P) at (90:10);
	\coordinate (V) at (210:10);
	\coordinate (U) at (230:10);
	\coordinate (M) at (0,0);
	\coordinate (HP) at  (190:10);
	\coordinate (HP2) at (-10:10);
	\draw[gray,dashed] (0,0) circle (10);

	\fill (P) circle (.25) node [above] {$p$};
	\fill (V) circle (.25) node [below left] {$y'$};
	\fill (U) circle (.25) node [below left]  {$y$};
	\fill (M) circle (.25) node [right] {$m$};
	
	\draw [name path=PV] (P) -- node [above] {} (V);
	\draw [name path=PU] (P) -- node [above] {} (U);
	\draw [name path=XY] (V) -- node [above] {} (U);
	\draw [name path=VM] (V) -- (0,0) -- (P);
	\draw [name path=HPP] (HP) -- (HP2)  node [below]  {$H_0$} ;
	
	\path [name intersections={of=HPP and PU,by=HPP-PY}];
	\draw (V) -- (U);
	\path [name path=XYSHIFT] (HPP-PY) -- +($(V)-(U)$);
	\path [name intersections={of=PV and XYSHIFT,by=PX-XYSHIFT}];
	\draw (PX-XYSHIFT) -- (HPP-PY);
	\path [name intersections={of=PV and HPP,by=HPP-PX}];
	\draw (HPP-PX) -- (HPP-PY);
	
	\fill (HPP-PX) circle (.25) 				node [ below]{$x'$};
	\fill (HPP-PY) circle (.25) 				node [below, xshift=6]{$x$};
	\fill ($(P)!(HPP-PY)!(V)$) circle (.25) 	node [above left]		{$b$};
	\fill (PX-XYSHIFT) circle (.25) 			node [above, xshift=2]		{$a$};
	
	\draw[name path=B] ($(P)!(HPP-PY)!(V)$) -- (HPP-PY);
	\draw[dotted, gray] (P) -- (HP)  node [midway,left ] {\tiny $=\sqrt{2}$};
	\draw[dotted, gray] (P) -- (HP2) node [midway,right] {\tiny $=\sqrt{2}$};
	\draw[name path=MH0, dotted, gray] ($(HP)!(M)!(HP2)$) -- (M);
	
	\draw (0,2) arc (90:210:2);
	\node[above left] (M) {\tiny $\gamma$};
	\end{tikzpicture}
\end{center}

In this section we denote with ${B=\overline{b x} }$ the perpendicular from $x$ on $\overline{py'}$,$E=\overline{xx'}$, ${L=\overline{yy'}}$ and 
${L_x=\overline{xa} }$ its triangle scaled version meeting $x$.
Note that $B$ and $L_x$ are above $H_0$, hence above $E$.

\begin{lemma} \label{lemInscrAngle}
	For $x,x' \in \B^{d-1}_1$ with ${\lVert x\rVert_2 \leq \lVert x'\rVert_2}$, we have 
	$$ 
	\frac{\lVert x \rVert_2}{\lVert p-\sigma(x) \rVert_2} \lVert \sigma(x) - \sigma(x') \rVert_2 \leq \lVert x - x' \rVert_2~.
	$$
\end{lemma}

\begin{proof}
	We show $L_x \leq E$ by proofing $\alpha \leq \beta$ for the two angles
	\begin{align*}
	\beta  &:= \measuredangle b x x'\\
	\alpha &:= \measuredangle a x b~.
	\end{align*}

	The inner angle sum of $\triangle xab$ with a supplementary angle argument and triangle scaling provide $\measuredangle py'y = 90^\circ + \alpha$.
	Let $m$ denote the center of $C$.
	Since $\overline{pm}$ is orthogonal on $H_0$ and $\measuredangle bx'x = 90^\circ - \beta$, we have $\measuredangle mp x' = \beta$.
	In the isosceles triangle $\triangle py'm$, the central angle $\gamma = 180^\circ - 2\beta$.
	Fixing arc $\overline{py'}$ on $C$ for the inscribed angle theorem provides $\measuredangle y'yp=\gamma/2$.
	
	Now, suppose $\alpha > \beta$.
	The inner angle sum of $\triangle pyy'$ states
	\begin{align*}
	0 \leq \measuredangle y'py
	&= 180^\circ - \measuredangle y'yp - \measuredangle py'y\\
	&= 180^\circ - \measuredangle y'yp - (90^\circ + \alpha)\\
	&= 180^\circ - \gamma/2  - (90^\circ + \alpha)\\
	& = -\alpha + \beta ~
	\end{align*}
	a contradiction.
\end{proof}

\begin{lemma} \label{lemGeoStretching}
	For ${x, x' \in \B^{d-1}_1 }$ , we have
	\begin{align*} 
	\Big\lVert\sigma( x) - \sigma(x') \Big\rVert_2 \leq 2 ~\lVert x - x' \rVert_2.
	\end{align*}
\end{lemma}
\begin{proof}
	Using Lemma \ref{lemInscrAngle}, we have $L_x\leq E$ and the statement follows via triangle scaling:
	$$
	L = L_x ~\overline{py} ~/~ \overline{px} \leq 2 L_x\leq  2 E ~,
	$$
	since $\overline{px} \geq 1$ and $\overline{py} \leq 2$.
\end{proof}
This statement is tight,
considering the two points $x = 0$ and
$x'=\left( \varepsilon/\sqrt{d-1}, \ldots, \varepsilon/\sqrt{d-1}\right)$.
We have $\lVert x - x'\rVert_2=\varepsilon$ and 
$
\left\lVert \sigma(x) - \sigma(x')\right\rVert_2
= 2\frac{1}{\sqrt{1+\varepsilon^2 }} \varepsilon
$.

\begin{theorem}\label{thmQuality}
Algorithm \ref{algoSnapping} calculates an $\varepsilon$-approximation exactly on the unit sphere.
\end{theorem}
\begin{proof}
Let $x^* = x /\lVert x \rVert_2$ and $\sigma(y)$ denote the result.
Given the rotation, $x^*$ holds for Observation \ref{obsTauInCircle}.
Hence, we can use Lemma \ref{lemGeoStretching} to derive
\begin{align*}
\lVert \sigma(y) - x^*\rVert_\infty
&= \lVert \sigma(y) - \sigma(\tau(x^*))\rVert_\infty\\
&\leq \lVert \sigma(y) - \sigma(\tau(x^*))\rVert_2 \\
&\leq 2 \lVert y - \tau(x^*) \rVert_2 \\
&\leq2\sqrt{(d-1) \frac{\varepsilon^2}{4(d-1)}} = \varepsilon 
\end{align*}
as upper bound on the approximation error.
\end{proof}
This analysis is rather tight, as demonstrated by the red curve and points in Figure \ref{figQualityAndSize3d}.

\subsection{Denominator Sizes}\label{secDenomSize}
We now describe a relation between rational images of $\sigma$ and the lowest common multiple of  denominators of its rational pre-images.
This leads to several strategies for achieving small denominators in the results of {Algorithm \ref{algoSnapping}}.
\begin{lemma}[Size of images under $\sigma$] \label{lemDenomSize}
Let $x \in {\Q^{d-1} \cap \B^{d-1}_1}$ with $x_i = p_i/q_i$ and $Q = lcm(q_1, \ldots, q_{d-1})$ be the lowest common multiple, then 
$$
\sigma_k \left(x\right)= \frac{n_k}{m}
$$ 
with integers $n_i, m \in \{-2Q^2, \ldots , 2Q^2\}$ for all $1\leq k \leq d$.
\end{lemma}
\begin{proof}
Let $q'_i \in \{1,\ldots,Q\}$ such that $q'_i \cdot q_i = Q$ for all $i$.
Since the formula of $\sigma$ is similar in all but the last dimension, we describe the following two cases.
For $k=d$, we have
\begin{align*}
\sigma_k \left(x\right)
&=\frac{-1 + \sum_{i=1}^{d-1}p_i^2 / q_i^2}{1 + \sum_{i=1}^{d-1}p_i^2 / q_i^2} 
=\frac{-Q^2 + \sum_{i=1}^{d-1}{q'_i}^2p_i^2 }{Q^2 + \sum_{i=1}^{d-1}{q'_i}^2p_i^2} =: \frac{n_k}{m}
\end{align*}
Using the bound $x \in \B^{d-1}_1$, we have $0 \leq \sum_{i=1}^{d-1}{q'_i}^2p_i^2 \leq Q^2$ and we derive for $n_k$ and $m$
\begin{align*}
	|n_k| &= \Big|-Q^2+ \sum_{i=1}^{d-1}{q'_i}^2p_i^2 \Big| \leq Q^2\\
	m &= Q^2 + \sum_{i=1}^{d-1}{q'_i}^2p_i^2 \leq 2Q^2
\end{align*}
	
For $k < d$, we have
\begin{align*}
	\sigma_k \left(x\right)
	&=\frac{2 p_k/q_k}{1 + \sum_{i=1}^{d-1}p_i^2 / q_i^2} \\
	&= \frac{Q^2 \cdot 2p_k/q_k }{Q^2 + \sum_{i=1}^{d-1}{q'_i}^2 p_i^2 } \\
	&=\frac{ Qq'_k \cdot 2 p_k }{Q^2 + \sum_{i=1}^{d-1}{q'_i}^2 p_i^2  }  =: \frac{n_k}{m}
\end{align*}
Using the bound $x \in \B^{d-1}_1$, we have that each $|p_i| \leq q_i $ and this bounds $|n_k| = Qq'_k \cdot 2 |p_k| \leq 2Q^2$.
We already discussed the bound on $m$ in the first case.
\end{proof}

Note that we apply this lemma in practice with fixed-point binary numbers $p_i/q_i \in P_s$.
Meaning all $q_i = 2^s = Q$ for some significant size $s$.
\begin{theorem}\label{thmAlgoDenSize}
	Denominators in $\varepsilon$-approximations of Algorithm \ref{algoSnapping} are at most 
	$$\frac{10(d-1)}{\varepsilon^2} ~.$$
\end{theorem}
\begin{proof}
Using standard multi-precision floating point arithmetics allows to derive rational values $y$, with denominators that are $Q=\lceil \frac{2\sqrt{d-1}}{\varepsilon}\rceil$.
Using $\varepsilon \leq 1/8$ and Lemma \ref{lemDenomSize} bounds the size of the denominators in images $\sigma$ with
\begin{align*}
2Q^2 &\leq 2\left(1+ \frac{2\sqrt{d-1}}{\varepsilon}\right)^2\\
&= \frac{2}{\varepsilon^2}\left( \underbrace{\varepsilon^2 +\varepsilon4\sqrt{d-1}}_{\leq (d-1)} + 4(d-1) \right)~.
\end{align*}
\end{proof}

For certain dimensions and in practice(c.f. Section \ref{secExpQualityAndSize}), we can improve on the simple usage of fixed-point binary numbers.
For $\S^1$ we can rely on the continued fraction algorithm to derive rational approximations of $\alpha = \tau(x/\lVert x\rVert_2)$ with $ \left| \alpha - p/q\right| \leq 1/2q^2$.
Using this in Algorithm \ref{algoSnapping} leads to approximations with $\varepsilon=1/q^2$ on the circle $\S^1$ with denominators of at most $2q^2$.

Note that for $\S^d$ with $d\geq 2$ one can rely on algorithms for simultaneous Diophantine approximations (c.f. Theorem \ref{thmDirichlet}) to keep the lowest common multiple $Q$ in Lemma \ref{lemDenomSize} small.
Note that it might well be \emph{simpler} to find Diophantine approximations with small $Q$.

There have been many approaches to find generalizations of the continued fraction algorithm for $d>1$.
One of the first approaches is the Jacobi-Perron algorithm, which is rather simple to implement\cite{tachii1994}(c.f. Section \ref{secExpQualityAndSize}).
More advanced approaches \cite{Pethoe2017422} rely on the LLL-algorithm for lattice basis reduction\cite{Lenstra82factoringpolynomials}.
For $d=2$ there is an algorithm to compute all Dirichlet Approximations\cite{Jurkat1979}, which we find hard to oversee given its extensive presentation.
Moreover, their experimental comparison shows that the Jacobi-Perron algorithm is practically well suited for $d=2$.

We close this section with a transfer result of Theorem \ref{thmDirichlet} with our Theorem \ref{thmQuality} and Lemma \ref{lemDenomSize}.
\begin{corollary}\label{corDirichletSnapping}
Let $x \in \S^{d-1}$ and $N \in \N$.
There is $p \in \Z^{d-1}$ and $q \in \{1,\ldots,N\}$ with  
$$
	\left\lVert x - \sigma\left(\frac{1}{q}p\right) \right\rVert_\infty \leq \frac{2\sqrt{d-1}}{q\sqrt[d-1]{N}}
$$
and all denominators of $\sigma\left(\frac{1}{q}p\right)$ are at most $2q^2$.
\end{corollary}
This existence statement allows for brute-force computations.
However, we just use it for comparisons in Section \ref{secExpQualityAndSize}.

\section{Implementation}
Apart from \cite{Teillaud2010} for $\S^2$, most implementations of spherical Delaunay triangulations are not `stable'.
Approaches based on $d$-dimensional convex hull algorithms produce only a tessellation for input not exactly \emph{on} $\S^{d-1}$. (c.f. Section \ref{secRelatedWork})

Few available implementations allow dynamic point or constraint insertion and deletion -- not even in the planar case of $\R^2$.
The `Computational Geometry Algorithms Library' (CGAL \cite{cgal47}) is, to our knowledge, the sole implementation providing dynamic insertions/deletions of points \emph{and} constraint line segments in $\R^2$. 

With \cite{libratss}, we provide open-source implementations of Algorithm \ref{algoSnapping} for $\S^d$. In \cite{libdts2}, we provide an implementation for spherical Delaunay triangulations on $\S^2$ with $\varepsilon$-stable constructions of intersection points of constraint line-segments (c.f. Section \ref{secEpsStableConstructions}).

\subsection{RATional Sphere Snapping for $\S^d$}
\href{http://www.github.com/fmi-alg/libratss}{Libratss} is a C++ library which implements Algorithm \ref{algoSnapping}, based on the open-source GMP library for exact rational arithmetics \cite{Granlund12} and the GNU `Multiple Precision Floating-Point Reliably'(MPFR) library\cite{Fousse:2007:MMB:1236463.1236468}.
The implementation allows both, input of Cartesian coordinates of arbitrary dimension and spherical coordinates of $\S^2$.
Note that this implementation allows geometric algorithms, as for $d$-dimen\-sional convex hull, to rely on rational input points that are \emph{exactly} on $\S^{d-1}$.
In light of the discussion on the denominator sizes in Section \ref{secDenomSize}, we provide two additional strategies to fixed-point snapping, as analyzed in Theorem \ref{thmQuality}.
We implemented the Continued Fraction Algorithm to derive rational $\varepsilon$-approximations with small denominators and the Jacobi-Perron algorithm for $\S^2$.
The library interface also allows to automatically chose the approximation method which results in smaller denominators, approximation errors or other objectives, like byte-size.

\subsection{Incremental Constrained Delaunay Triangulation on $\S^2$}
\href{http://www.github.com/fmi-alg/libdts2}{Libdts2} implements an adapter for the \emph{dynamic constraint} Delaunay triangulation in the Euclidean plane $\R^2$ of CGAL.
Since this implementation requires an initial outer face, we introduce an small triangle, that only contains the north-pole, to allow subsequent insertions of points and constraints.
For points \emph{exactly} on the unit sphere, the predicate \textbf{`is $A$ in the circumcircle of $B,C$ and $D$'} reduces to the well studied predicate \textbf{`is $A$ above the plane through $B,C$ and $D$'}.
The implementation overloads all predicate functions accordingly and uses Algorithm \ref{algoSnapping} for the construction of rational points on the sphere for intersections of Great Circle segments.

\subsubsection{$\epsilon$-stable geometric constructions} \label{secEpsStableConstructions}
Any means of geometric construction that allows to approximate a certain point, can be used as input for Algorithm \ref{algoSnapping} -- E.g. the intersection of Great Circle segments.
Consider two intersecting segments of rational points on $\S^2$. 
The two planes, containing the segments and the origin as a third point, intersect in a straight line.
Each (rational) point on this line can be used as input for our method, as they identify the two intersection points on the sphere.
Using such input for {Algorithm~\ref{algoSnapping}} allows simple schemes to derive stable geometric constructions of rational points on $\S^d$ within a distance of $\varepsilon$ to the target point.

\section{Experiments}
We used real world and synthetic data for our experiments.
Geo-referenced data was sampled from regional extracts from the OpenStreetMap project\cite{osm}, as of January $26$th, 2017.
Random Cartesian coordinates of points on $\S^d$ were created with the uniform generator 2 of \cite{marsaglia1972choosing}.
All benchmarks were conducted on a single core of an Intel Xeon E5-2650v4.
Peak memory usage and time were measured using the \verb|time| utility.

\begin{figure}
	\centering
	\includegraphics[width=\columnwidth]{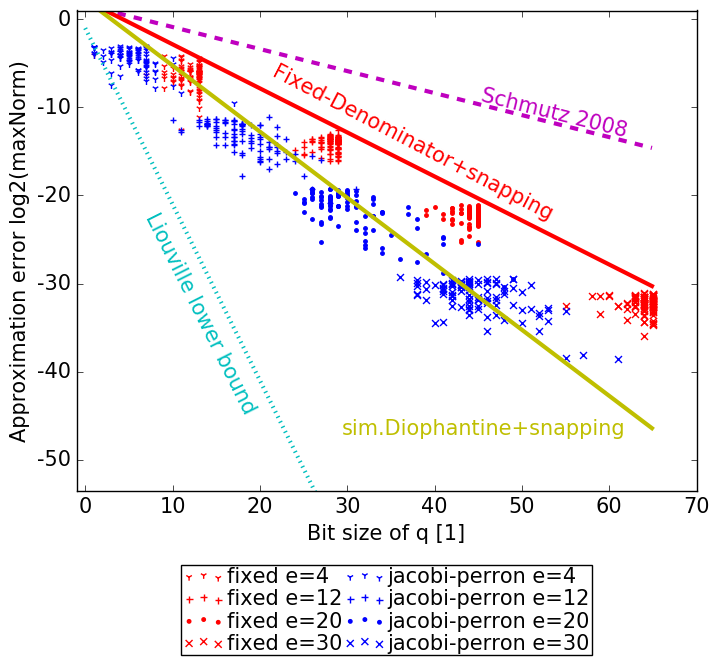} 
	\vspace{-.7cm}
	\caption{Approximation quality and denominator size of $100$ random points on $\S^2$ for various levels of target precision $e$ and approximation strategies (red, blue) of Algorithm \ref{algoSnapping}.
		Theoretic bounds are indicated with lines.} \label{figQualityAndSize3d}
	\vspace{-.1cm}
\end{figure}

\subsection{Approximation Quality and Size} \label{secExpQualityAndSize}
We experimentally analyze the actual approximation error in results of Algorithm \ref{algoSnapping} for several levels of $\varepsilon$ using the MPFR library.
In this section $e$ denotes the significands required in statement \ref{algoStatementChoose} of Algorithm \ref{algoSnapping} for the required result precision $\varepsilon$. This is 
$$
e = \left\lceil -\log_2 \left( \frac{\epsilon}{2\sqrt{d-1}} \right) \right\rceil \quad.
$$
We simply setup the MPFR data types with significand sizes up to $1024$ Bits, and conducted our experiments on much lower levels of $e$.
This allows us to derive some `measure' of the actual approximation errors of our method.

We analyzed the approximation errors $\delta$ and denominator bit-sizes $q$ for 
$100$ random points on $\S^2$. 
Figure \ref{figQualityAndSize3d} compares the results of our algorithm under several levels of target precision $e$ and strategies for statement \ref{algoStatementChoose} in our method.
The magenta line indicates the quality and size of the approach in \cite{Schmutz2008}.
The red line indicates the bounds of our Theorems \ref{thmQuality} and \ref{thmAlgoDenSize} on the fixed-point strategy, while the yellow line indicates the bound of Corollary \ref{corDirichletSnapping}.
Note that results using the Jacobi-Perron strategy (blue dots) allows our method to further improve on the fixed-point strategy (red dots).
Note that we use Liouville's lower bound as statement on the approximability of a worst-case point.
There might well be points of higher algebraic degree that allow better approximations (c.f. Section \ref{sectLowerBounds}).

Table \ref{table:snappingToRational} exhibits average approximation errors $\delta$, denominator bit-sizes $q$ and the computation time $t$ of our method for millions of points.
Synthetic data sets have several dimensions, while the real world data sets have dimension $3$.
For $\S^2$, we provide comparison of the fixed-point strategy (fx) with the Jacobi-Perron strategy (jp) of our method.
Using $e=31$ is sufficient to obtain results \emph{exactly on} $\S^2$ with a $\delta$ of less than $1$cm, relative to a sphere with radius of the earth.
This is enough for most applications dealing with spatial data \emph{and} allows storage within the word size of contemporary computing hardware.
This allows practical applications on $\S^2$ to store $4$ integer long values for the $3$ numerators and the common denominator (c.f. Lemma \ref{lemDenomSize}) occupying $32$ Bytes.
Note that storing $3$ double values occupies $24$ Bytes but \emph{cannot} represent Cartesian coordinates \emph{exactly on} the sphere.

\begin{table}
\small
\setlength{\tabcolsep}{0.3em} 
\begin{tabular}{rr|r|r|r|r|r}
&&Germany&Planet&u.a.r $\mathbb{S}^2$&u.a.r $\mathbb{S}^{9}$&u.a.r $\mathbb{S}^{99}$\\
\multicolumn{2}{c|}{dimension}&3&3&3&10&100\\
\multicolumn{2}{c|}{size $[10^3]$}&2,579.6&3,702.4&1,000.0&1,000.0&100.0\\
\hline
\multicolumn{7}{c}{$e$=23}\\
\hline
\multirow{3}{*}{fx}&$\delta [ m ]$&0.7&0.7&0.7&1.0&3.2\\
&q $[ 1 ]$&46.0&46.0&46.0&46.0&46.0\\
&t $[ \mu s ]$&17&16&16&117&546\\
\cline{1-7}
\multirow{3}{*}{jp}&$\delta [ m ]$&0.4&0.4&0.5&-&-\\
&q $[ 1 ]$&33.6&34.2&34.1&-&-\\
&t $[ \mu s ]$&63&57&58&-&-\\
\hline
\multicolumn{7}{c}{$e$=31}\\
\hline
\multirow{3}{*}{fx}&$\delta [ m ]$&2.7e-3&2.6e-3&2.8e-3&4.0e-3&12.6e-3\\
&q $[ 1 ]$&62.0&62.0&62.0&62.0&62.0\\
&t $[ \mu s ]$&17&16&17&118&554\\
\cline{1-7}
\multirow{3}{*}{jp}&$\delta [ m ]$&1.7e-3&1.7e-3&1.8e-3&-&-\\
&q $[ 1 ]$&45.2&45.8&45.8&-&-\\
&t $[ \mu s ]$&77&72&73&-&-\\
\hline
\multicolumn{7}{c}{$e$=53}\\
\hline
\multirow{3}{*}{fx}&$\delta [ m ]$&6.3e-10&6.2e-10&6.6e-10&9.6e-10&30.1e-10\\
&q $[ 1 ]$&106.0&106.0&106.0&106.0&106.0\\
&t $[ \mu s ]$&16&16&17&118&548\\
\cline{1-7}
\multirow{3}{*}{jp}&$\delta [ m ]$&3.9e-10&3.9e-10&4.3e-10&-&-\\
&q $[ 1 ]$&77.2&77.8&77.7&-&-\\
&t $[ \mu s ]$&118&111&112&-&-\\
\hline
\multicolumn{7}{c}{$e$=113}\\
\hline
\multirow{3}{*}{fx}&$\delta [ m ]$&5.5e-28&5.4e-28&5.7e-28&8.3e-28&26.1e-28\\
&q $[ 1 ]$&226.0&226.0&226.0&226.0&226.0\\
&t $[ \mu s ]$&19&19&19&126&617\\
\cline{1-7}
\multirow{3}{*}{jp}&$\delta [ m ]$&3.4e-28&3.4e-28&3.7e-28&-&-\\
&q $[ 1 ]$&164.5&165.1&165.1&-&-\\
&t $[ \mu s ]$&219&218&220&-&-\\
\end{tabular}
\caption{Mean-values of approximation error $\delta~[m]$, denominator bit-size $q~[1]$ and computation time $t~[\mu s]$ for synthetic and real-world point sets for various dimensions and levels of target precision $e$.
The Jacobi-Perron strategy is denoted by `jp' and the fixed-point strategy by `fx'.}
\label{table:snappingToRational}
\end{table}

\begin{table}[]
	\centering
	\setlength{\tabcolsep}{0.3em} 
	\begin{tabular}{r|rrrr}
		&Saarland	&Germany	&Europe		&Planet \\
		\hline
		\multicolumn{5}{c}{Input} \\
		\hline
		Segments $\left[10^6\right]$	&$0.32$		&$25.75$		&$222.92$		&$668.61$\\
		\hline
		\multicolumn{5}{c}{Output}\\
		\hline
		Vertices $\left[10^6\right]$	&$0.29$	&$24.45$	&$213.01$	&$634.42$\\
		Edges $\left[10^6\right]$		&$0.87$	&$73.37$	&$639.04$	&$1,903.27$\\
		Faces $\left[10^6\right]$		&$0.58$	&$48.91$	&$426.03$	&$1,268.84$\\
		\hline
		\multicolumn{5}{c}{Resource usage}\\
		\hline
		Time	[h:m]		&$<0$:$01$	&$19$:$27$	&$3$:$21$	&$12$:$04$\\
		Memory	[GiB]		&$0.3$		&$20.4$		&$182$		&$545$\\
	\end{tabular}
	\caption{Time and memory usage to compute spherical Delaunay triangulations for OpenStreetMap data sets.}\label{tab:dtbenchmark}
\end{table}

\subsection{Constrained Delaunay Triangulation with Intersection Constructions} \label{secConstrDelaunayWithInter}
A Constrained Delaunay Triangulation of a point set contains required line-segments as edges, but is as close to the Delaunay triangulation as possible \cite{Chew:1987:CDT:41958.41981}.
We used \emph{very large} street networks of several regions from the OpenStreetMap project for points and constraint edges -- E.g. each line-segment of a street is an edge in the result triangulation.
Since $\sim0.5\%$ of the line-segments in these data sets intersect, we approximated the intersection points using $e=31$ for Algorithm \ref{algoSnapping}.
Table \ref{tab:dtbenchmark} exhibits total running time, peak memory usage and the result sizes of our \verb|libdts2| implementation.
Small data sets like Saarland and Germany allow quick calculation on a recent workstation computer.
See Figure \ref{figMotivation} for the Ecuador dataset.
Note that the current implementation has a storage overhead for each point, as we keep the results of the GMP library rather than truncating to integers of architectures word size.
Computing the triangulation for the planet data set was only possible on rather powerful hardware with at least $550$ Gigabytes of memory taking half a day.

\section{Open Problems}
From a practical point of view, it is of great interest to bound the storage size of denominators to a maximum of $64$Bits -- the word size of current computing architectures.
We seek to improve our (already satisfactory) results by using advanced algorithms for simultaneous approximation, like the LLL-algorithm or the Dirichlet approximation algorithm for $\S^2$.

For the theoretical part, we are interested if finding simultaneous rational approximations with small lowest common multiple of the denominators is simpler than finding Dirichlet approximations.
We are also interested in generalizing the method to provide rational approximations with small \emph{absolute} errors on ellipsoids with rational semi-principal axes -- E.g. the geographic WGS84 ellipsoid.

\clearpage
	
	\bibliographystyle{abbrv} 
	\bibliography{./ref/sphericalDelaunay.bib}
	
	\clearpage
	\appendix

	\section{Proof of Liouville's Approximation Theorem} \label{apxProofLiouville}
	This nice proof was translated from the German  \href{https://de.wikibooks.org/wiki/Beweisarchiv:_Algebra:_K%C3%B6rper:_Approximationssatz_von_Liouville}{wikibooks} Project -- thanks to the anonymous authors.
		See \cite{1851-Liouville} for the original proof in French language.
		
		\numberwithin{equation}{section}
		
		\begin{proof}[of Theorem \ref{thmLiouville}]
			Let $\alpha \in \R$ be algebraic of degree $n$ and root of the corresponding polynomial $f(X)\in\Z[X]$ of degree $n$, meaning
			\begin{align*}
			f(\alpha) = a_0 + a_1\alpha + \cdots +a_n\alpha^n = 0
			\end{align*} with $a_0, \ldots, a_n \in \Z$ and $a_n\ne 0$.
			
			Polynomial division with the linear factor $X-\alpha$ in the ring $\C[X]$ provides
			\begin{align}
			f(X) = (X-\alpha)\cdot g(X).	\label{apx:eqn:polyDivision}
			\end{align}
			
			Note that the polynomial $g(X)$ has algebraic coefficients and is not necessarily in $\Z[X]$.
			However, the mapping ${\R\to\C}$, $t\mapsto g(t)$ is continuous, by means of real numbers $c_1>0$, $c_2>0$ with
			\begin{align}
			\left|g(x)\right| \leq c_1	\label{apx:eqn:continous}
			\end{align}
			for $|\alpha-x| < c_2$.
			Since $n<\infty$, we can assume w.l.o.g. that no additional roots are in this neighborhood of $\alpha$, meaning
			\begin{align}
			f(x)\ne 0	\label{apx:eqn:emptyInterval}
			\end{align}
			for $|\alpha-x| < c_2$ and $x\ne \alpha$.
			
			Claim: The statement of the Theorem holds for
			\begin{align*}
			c:=\min\left\{c_2, \frac{1}{c_1}\right\} \quad.
			\end{align*}
			
			Suppose there are $p,q\in\Z$, $q>0$ with
			\begin{align}
			\left|\alpha-\frac pq\right | < \frac c{q^n} \quad. \label{apx:eqn:goodApx}
			\end{align}
			We show that his implies $\alpha=\tfrac pq$.
			
			From (\ref{apx:eqn:goodApx}), we immediately derive
			\begin{align}
			\left|\alpha-\frac pq\right| < c \leq c_2 ~, \label{apx:eqn:start}
			\end{align}
			leading (\ref{apx:eqn:continous}) to imply $\left| g(\tfrac pq) \right | \leq c_1$.
			We derive, from (\ref{apx:eqn:polyDivision}) and again (\ref{apx:eqn:goodApx}), that
			\begin{align*}
			\left|f\Bigl(\frac pq\Bigr) \right| = \left|\frac pq -\alpha\right|\cdot \left|g\Bigl(\frac pq\Bigr)\right| < \frac {c}{q^n} \cdot c_1 \leq \frac 1{q^n},
			\end{align*}
			meaning $\left|q^n \cdot f\Bigl(\frac pq\Bigr)\right| < 1.$
			
			However $q^n\cdot f\Bigl(\frac pq\Bigr) = a_0q^n + a_1pq^{n-1} + \cdots + a_n p^n\in\Z$ and its absolute value is smaller than $1$, hence has to be $0$.
			Moreover, $f(\tfrac pq) = 0$ and (\ref{apx:eqn:start}) with (\ref{apx:eqn:emptyInterval}) imply $\alpha=\tfrac pq$, which closes the argument.
		\end{proof}
		
		\section{Proof of Dirichlet's Approximation Theorem}\label{apxProofDirichlet}
		The folklore proof bases on Dirichlet's famous Pigeonhole Principle.
		See \href{https://proofwiki.org/wiki/Dirichlet%27s_Approximation_Theorem}{proofwiki.org} or Chapter 11.12 in \cite{hardy54}.
			\begin{proof}[of Theorem \ref{thmDirichlet}]
				We consider the partition of $[0,1]^d$ in $N^d$ regular $d$-cubes of length $L=\sqrt[d]{N}$.
				We further define a sequence of points $(a^{(j)})_{j=1,\ldots, N^d+1} \in [0,1]^d$ with $a^{(j)} := j \cdot \alpha - \lfloor j \cdot \alpha \rfloor $ (component wise operations).
				There are indices $k > l$ such that the points $a^{(k)}$ and $a^{(l)}$ are contained in the same $d$-cube.
				We have the (component wise) inequalities
				\begin{align*}
				-\frac{1}{L} < a^{(k)} &- a^{(l)} < \frac{1}{L} \\
				-\frac{1}{L} < k\alpha - \lfloor k \alpha \rfloor &- l\alpha+\lfloor l \alpha \rfloor < \frac{1}{L} \\
				-\frac{1}{L} < (k-l)\alpha &- (\lfloor k \alpha \rfloor  - \lfloor l \alpha \rfloor ) < \frac{1}{L} \\
				\end{align*}
				Setting $q=k-l$ and $p_i = \lfloor k\alpha_i\rfloor - \lfloor l\alpha_i\rfloor$ provides integers as required.
			\end{proof}

			\section{Reductions of Spherical Predicates to Cartesian Orientation Predicates}\label{apxPredRed}
			We first describe a reduction from the spherical predicates to well studied Cartesian predicates.
			
			\begin{lemma}[Great Circle Orientation Predicate]
				Let $p_1, p_2\in \mathbb{S}^2$ with $p_1 \neq p_2$ and $P$ the plane containing $p_1,p_2$ and the origin $(0,0,0)$ and $C$ be the Great Circle through $p_1$ and $p_2$. For $q \in \mathbb{S}^2$ we have 
				\begin{align*}
				q \text{ left-of } P &\iff q \text{ left-of } C\\
				q \in P &\iff q \in C\\
				q \text{ right-of } P &\iff q \text{ right-of } C\\
				\end{align*}
			\end{lemma}
			\begin{proof}
				$\mathbb{S}^2 \cap P = C$ and $\mathbb{S}^2= L \cup C \cup R$.
			\end{proof}
			
			\begin{lemma}[Circumsphere Predicate]
				Let $P$ denote the plane through non-identical points $p_1, p_2, p_3 \in \mathbb{S}^2$ and the half space containing the origin $(0,0,0)$ is called `below $P$'.
				We further call $S_{123} \subseteq \R^3$ the closed volume of the sphere with $p_1,p_2,p_3$ and the origin on its surface. For a point $q \in \mathbb{S}^2$ we have
				\begin{align*}
				q \text{ above } P &\iff  q \in S_{123} \setminus \partial S_{123}	\\
				q \in P 		   &\iff  q \in \partial S_{123} 					\\
				q \text{ below } P &\iff  q \notin S_{123} 							
				\end{align*}
			\end{lemma}
			
			\begin{proof}
				$P$ is uniquely determined because three different points on the unit sphere are not co-linear.
				Since $S_{123}$ and $\mathbb{S}^2$ are spheres, their cuts with $P$ are circles in $P$ and the two circles are identical as they contain $p_1,p_2$ and $p_3$ on their boundary.
				This circle $C$ has a radius of at most $1$ and partitions the points of the unit sphere into three sets $$\mathbb{S}^2 = A \cup C \cup B$$ where $A \subseteq S_{123} \supsetneq B$.
				If $C$ is a Great Circle we resolve ambiguity for `above' and the center of $S_{123}$ by choosing the open half spaces that first contain $(0,0,1)$, then $(0,1,0)$ and eventually $(1,0,0)$.
				We have $\mathbb{S}^2 \neq \partial S_{123}$, since the origin is a fourth point on $S_{123}$ and $q \in P$ iff. $q \in C$ iff. $q \in \partial S_{123}$.
				Therefore it is sufficient to show $q \in A	\iff q \text{ above } P $.
				To this end we consider the convex volume of the unit sphere $S \subseteq \R^3$ and $ D = S \cap S_{123}$. Note that $\partial D$ contains $C$ and $A$.
				Since the cut with the closed half-space of the plane $P$ cuts a convex body into at most three parts and $C \subseteq P$, we have that all of $A$ is `above' $P$.
			\end{proof}
			
			\begin{figure*}[b]
				\centering
 				\includegraphics[clip,width=1.7\columnwidth]{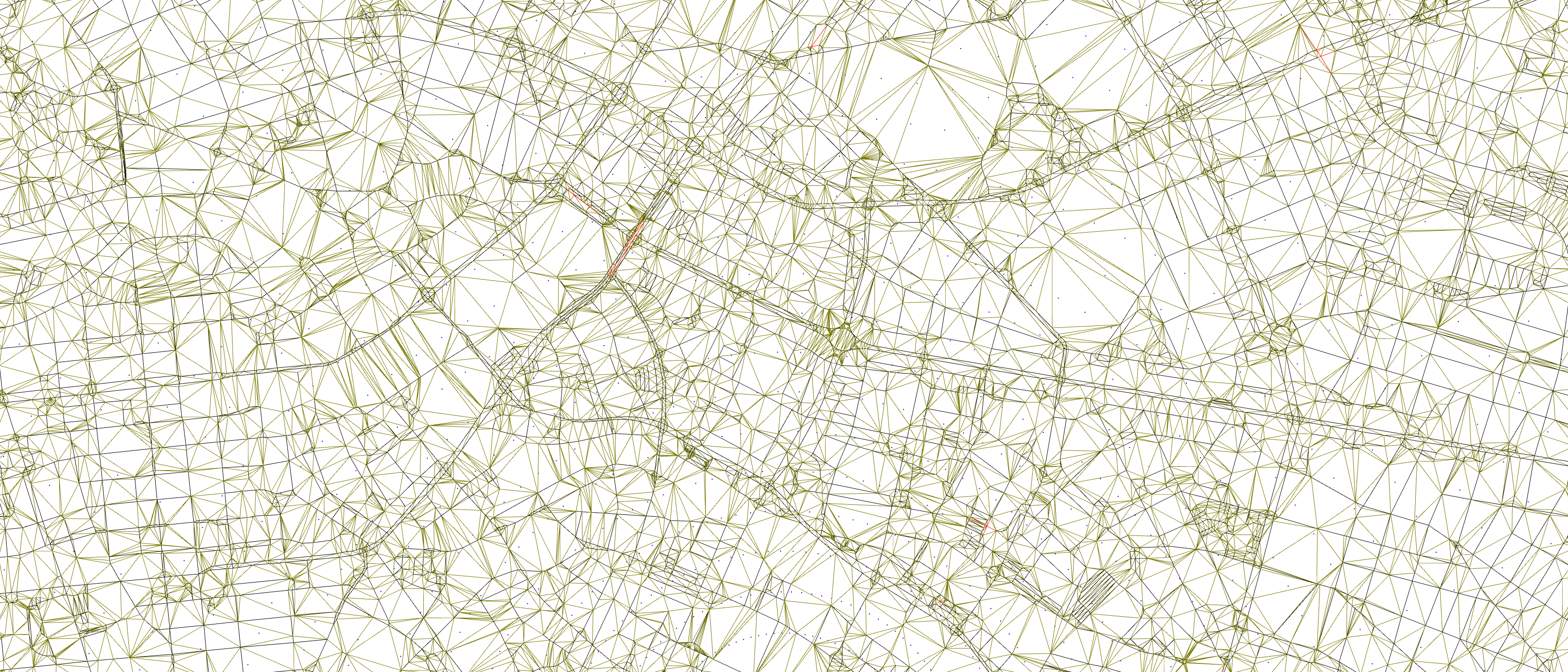}\\
 				\includegraphics[clip,width=1.7\columnwidth]{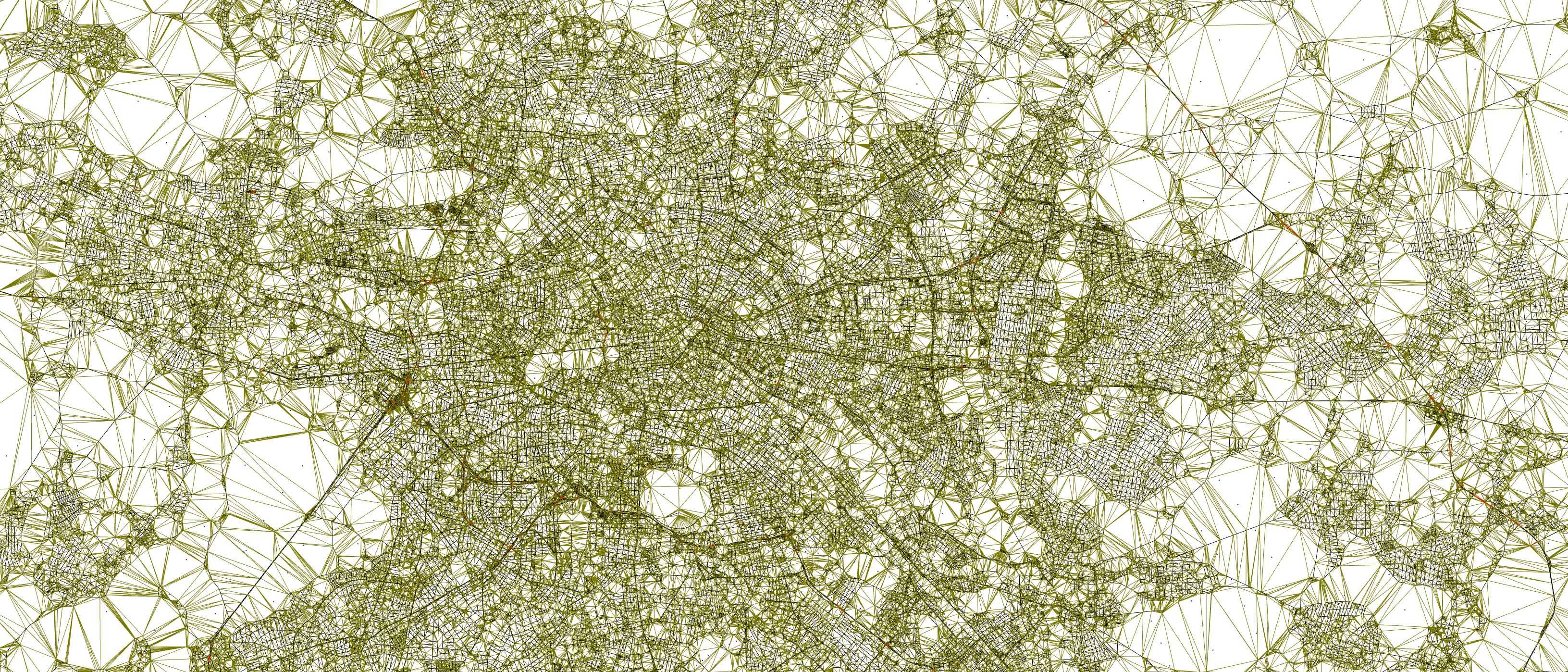}\\
 				\includegraphics[clip,width=0.9\columnwidth]{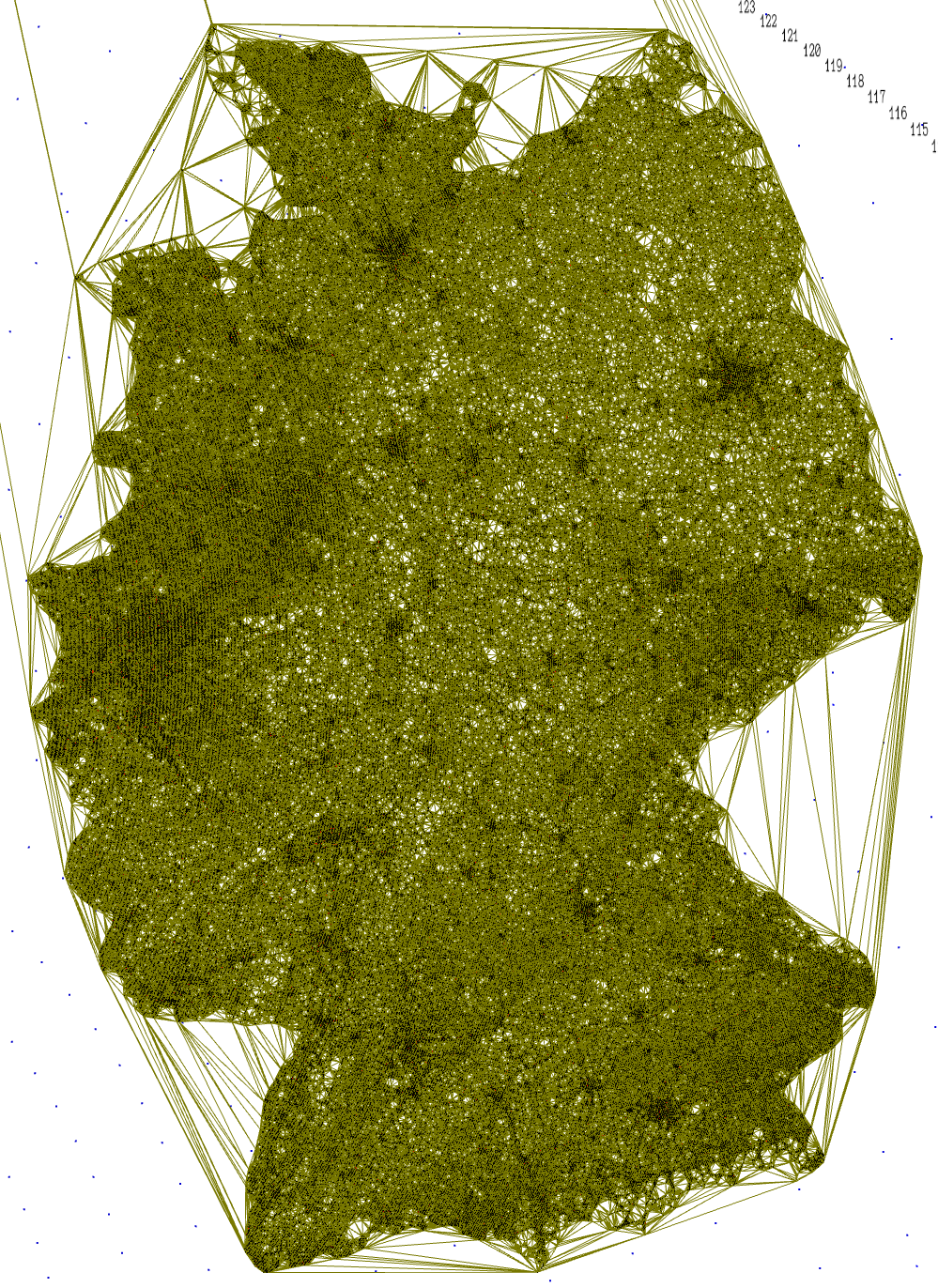}
 				\includegraphics[clip,width=0.9\columnwidth]{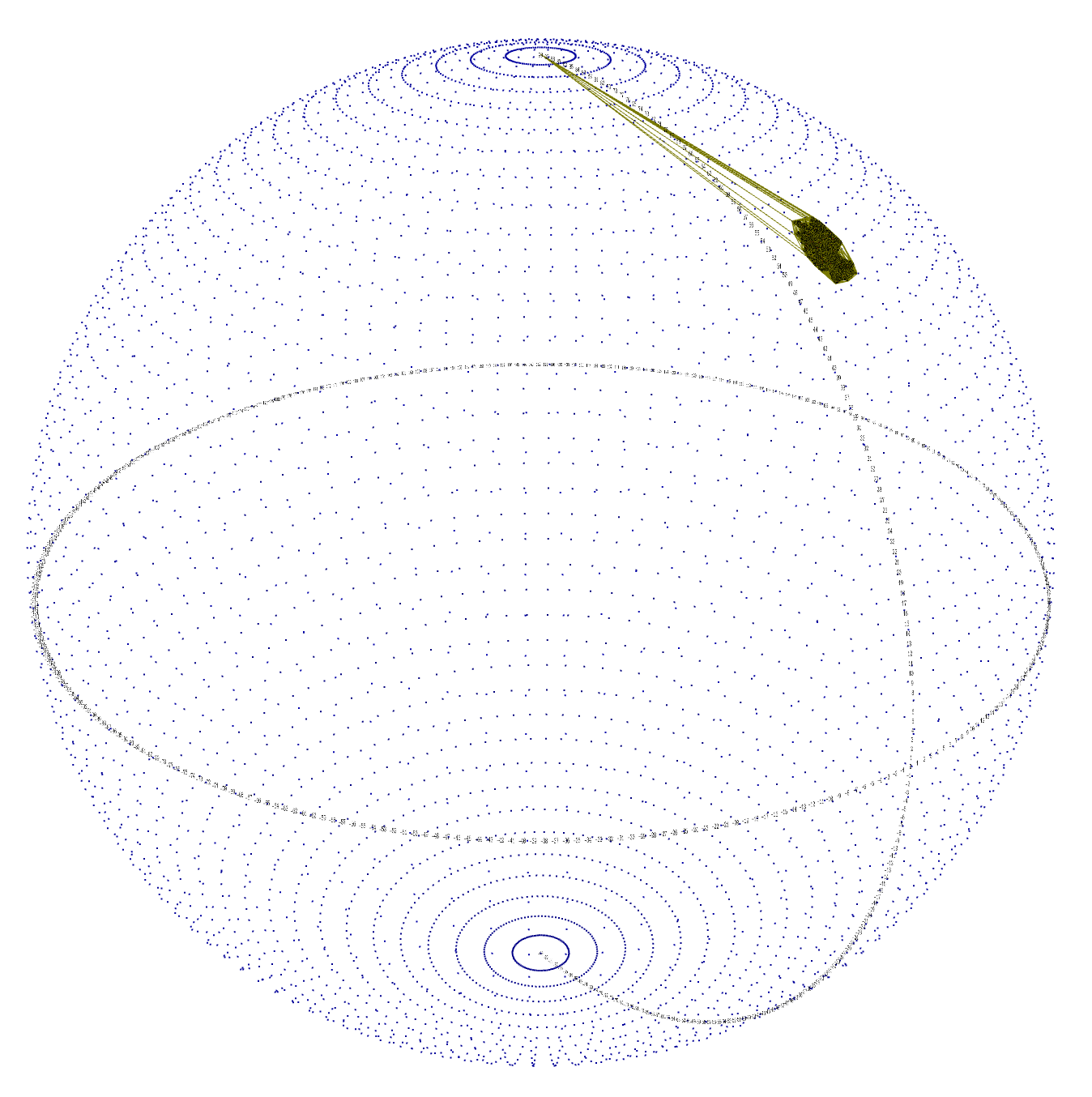}
				\caption{Spherical constraint Delaunay triangulation(green) of all streets(black) in the Germany data set (c.f. Section \ref{secConstrDelaunayWithInter}).}\label{fig:dtimages}
			\end{figure*}

		\end{document}